\documentclass{article}

\usepackage{PRIMEarxiv}

\usepackage[utf8]{inputenc} 
\usepackage[T1]{fontenc}    
\usepackage{hyperref}       
\usepackage{url}            
\usepackage{booktabs}       
\usepackage{amsfonts}       
\usepackage{nicefrac}       
\usepackage{microtype}      
\usepackage{lipsum}
\usepackage{fancyhdr}       
\usepackage{graphicx}       
\graphicspath{{media/}}     
\usepackage{graphicx}
\usepackage{amssymb}
\usepackage{amsmath}
\usepackage{tikz}
\usepackage{amsthm}
\usepackage{braket}

\newtheorem{thm}{Theorem}

\newtheorem{definition}{Definition}
\newtheorem{exmp}{Example}

\newtheorem{remark}{Remark}
\usepackage{color}
\usepackage{dutchcal}

\title{On Quantum Context-Free Grammars
}

\author{
  Merina Aruja \\
  Saintgits College of Engineering, Kottayam, Kerala, India \\ Amal Jyothi College of Engineering, Kanjirappally, Kerala, India\\
  APJ Abdul Kalam Technological University,Thiruvananthapuram,
  India\\
  \texttt{\{merina.aruja@saintgits.org\}} \\
   \And
 Lisa Mathew$^{* }$, Jayakrishna Vijayakumar \\
  Amal Jyothi College of Engineering,  Kanjirapally, Kerala,
  India\\
\texttt{\{lisamathew@amaljothi.ac.in, vjayayakrishna@amaljothi.ac.in\}} \\
}

\begin{document}
\maketitle

\begin{abstract}
Quantum computing  is a relatively new field of computing, which utilises the fundamental concepts of quantum mechanics to process data. The seminal paper of Moore et al. [2000] introduced quantum grammars wherein a set of amplitudes was attached to each production. However they did not study the final probability of the derived word. Aruja et al. [2025] considered conditions for the well-formedness of  quantum context-free grammars (QCFGs), in order to ensure that the probabilty of the derived word does not exceed one.
In this paper we propose certain necessary and sufficient conditions (also known as unitary conditions) for well-formedness of QCFGs.
\end{abstract}


\section{Introduction}

 Feynman's groundbreaking idea in 1981 \cite{Feynman} of using computers to simulate physical phenomena marked the origins of Quantum computing. The reversible unitary evolution of a quantum  process was simulated by Benioff \cite{Benioff} in 1982 using the quantum mechanical Hamiltonian model of a Turing machine. In 1985 Deutsch  \cite{Deutsch} presented a comprehensive quantum model for computation, which described a universal quantum system that could perfectly simulate any finite realizable physical system. In 1994 Shor \cite{Shor} presented polynomial time algorithms for two problems known to be in NP - integer factorisation and the discrete logarithm. 
\par 
Quantum automata and grammars have been widely studied \cite{Krishnamurty, Malyshev, Jismi, Wei&Jue} since the landmark study by Moore and Crutchfield \cite{moore2000quantum} extended ideas from classical automata to their quantum equivalents. In \cite{moore2000quantum} quantum variants of finite state automata and push-down automata were presented. Regular grammars and context-free grammars were converted into their quantum counterparts. Grammars were converted to their quantum counterparts by associating an amplitude vector with each production.

\par   The generalized real-time quantum automata and grammars discussed in \cite{moore2000quantum} were not inherently constrained to satisfy the condition that the language generated is well-formed i.e. the probability associated with the generation of a  word lies between $0$ and $1$. To overcome this problem we had formulated in \cite{Meri} certain sufficient conditions to ensure the property of well-formedness.
\par The quantum grammars  defined in \cite{moore2000quantum} and \cite{Meri} did not satisfy the condition of unitarity.
In fact, \cite{moore2000quantum} thoughtfully observes, "It is not clear what constraints a quantum grammar must satisfy to be unitary." Building on this insightful observation, we are motivated to delve deeper into the foundational aspects of quantum grammars and rigorously identify the conditions under which a quantum context-free grammar can exhibit unitarity. Hence in this paper we propose certain conditions which guarantee the unitary nature of quantum context free grammars.
 

\section{Preliminaries}
\subsection{Quantum Grammars}
As per Moore and Crutchfield \cite{moore2000quantum} a quantum grammar is defined as follows
\begin{definition}\cite{moore2000quantum}
 A quantum grammar $G = (N,T,I,P)$ consists of two alphabets $ N$ and $T$ (the set of nonterminals and the set of  terminals respectively), an initial variable $I \in N$  and a finite set $P$ of productions $\alpha \rightarrow \beta$ where
$\alpha \in N^*$ and $\beta \in {(N\cup T)^*}$.
Each production in $P$ has a set of complex amplitudes $c_k(\alpha \rightarrow \beta)$, $1\leq k \leq n$ where $n$ is the dimensionality of the quantum grammar.
\end{definition}
\begin{remark}
    This definition associates a vector $\textbf{c}(A, \beta, G)$  (or $\textbf{c}(p,G)$) of dimension $n$  to the production  $p:A \rightarrow \beta$   . 
    \end{remark}

\par Given strings $w_1$, $w_2$ and a production rule $p: A \rightarrow \gamma $, the vector $\textbf{c}(p)$ was associated with the derivation $w_1Aw_2 \implies w_1\gamma w_2$.  Consider the derivation 
\[\alpha = \alpha_0 \overset{p_1}{\implies} \alpha_1 \overset{p_2}{\implies} \alpha_2 \cdots \overset{p_m}{\implies}\alpha_m=\beta\] with amplitude vector $<c_{i_1}, c_{i_2}, \cdots c_{i_n}>$ corresponding to the $i^{th}$  step, $\alpha_{i-1} \implies \alpha_i$  of the derivation.
Then, the $k^{th}$ amplitude $c_k$ of the derivation   $\alpha \overset{p_1p_2\cdots p_m}\implies \beta$ was defined as the product $\prod_{i=1}^n c_{i_k} $ of the $k^{th}$ amplitudes for each production in the chain and $c_k(\alpha \overset{*}\implies \beta)$ as the sum of the $c_k^,s$ of all possible derivations
of $\beta$ from $\alpha$. The $k^{th}$ amplitude of a word $w \in T^*$ is 
 $c_k(w) = c_k(I \overset{*}\implies w)$  and the probability $f(w)$ associated with the word $w$ is given by
 \[f(w)=\sum^n_{k=1}|c_k(w)|^2\] Then $G$ generates the
quantum language $f$.\\
 The following notations were introduced in \cite{Meri}:
\begin{enumerate}
    
    \item The vector of amplitudes of the derivation  given by \[\alpha_0\implies\alpha_1\implies\cdots\implies\alpha_m\] is denoted by  $\textbf{c}(\alpha_0,\alpha_m) $.
    \item The amplitude of the word $w$ is denoted $\textbf{c}(w)$.
     \item $R(A)=\{\beta \in  (N\cup T)^*|A\rightarrow\beta\}$ for $A \in N$ .
      
\end{enumerate}
\noindent Moreover the concept of amplitude was extended to words over $(N\cup T)^*$  . 
\begin{definition}\cite{moore2000quantum}
A quantum grammar is context-free (QCFG) if the productions are of the form $\alpha \rightarrow \beta$, $\alpha \in N , \beta \in (N \cup T)^*$.
\end{definition} 
  We now recall some notions from functional analysis and quantum mechanics for completeness of our discussion.
 \begin{definition}
    The inner product of two vectors $u= <c_1, c_2, \cdots c_n>$ and $v=<d_1, d_2, \cdots d_n>$ is defined as  \[ \langle u, v \rangle = \sum_{i=1}^{n} \bar{c_i}d_i\].
\end{definition}
\begin{definition}
    Two vectors $<c_1, c_2, \cdots c_n>$ and $<d_1, d_2, \cdots d_n>$  are said to be orthogonal if their inner product is zero i.e\[\sum_{i=1}^{n} \bar{c_i}d_i=0\].
\end{definition}
\begin{definition}\label{Uni}
A unitary matrix is a square matrix $U$ of complex numbers that satisfies the condition
\[ U^{-1} = U^* \]
where $U^*$ denotes the conjugate transpose of $U$. Alternatively, a matrix $U$ is unitary if
\[ UU^* = U^*U = I \]
where $I$ is the identity matrix of the same order as $U$.
\end{definition}


    

In quantum mechanics entanglement occurs when two or more quantum systems become correlated in such a way that the state of one system cannot be described independently of the state of the other(s). Entangled states exhibit correlations that are stronger than those allowed by classical physics.\\
Thus states can be classified based on their representation in Hilbert spaces. A quantum state is considered separable if it can be expressed as a product of states of individual systems. For example, for two subsystems \( A \) and \( B \), a separable state can be written as:

\[
\ket{\psi}_{AB} = \ket{\psi_A} \otimes \ket{\psi_B}
\]

where \( \ket{\psi_A} \) and \( \ket{\psi_B} \) are states of subsystems \( A \) and \( B \), respectively.
In contrast, an entangled state cannot be factored into a product of individual states. An example of an entangled state for two qubits is:
\[
\ket{\psi}_{AB} = \frac{1}{\sqrt{2}} \left( \ket{00} + \ket{11} \right)
\]
This state exhibits correlations between measurements on subsystems A and B that cannot be explained by any local hidden variable theory. The overall state space of a composite system formed from two subsystems A and B is given by the tensor product of their respective Hilbert spaces:
The combined Hilbert space for two subsystems \( A \) and \( B \) can be expressed as:

\[
\mathcal{H}_{AB} = \mathcal{H}_A \otimes \mathcal{H}_B
\]

If both \( \mathcal{H}_A \) and \( \mathcal{H}_B \) are separable Hilbert spaces, then \( \mathcal{H}_{AB} \)
  is also separable. This structure allows for the mathematical treatment of entangled states within the framework of separable spaces.
\par The separability criterion provides a clear mathematical distinction between separable and entangled states. If the density matrix cannot be expressed as a convex combination of product states, it indicates entanglement. The separability of Hilbert spaces impacts how measurements on entangled systems are interpreted. When measuring one subsystem in an entangled pair, the outcome instantaneously influences the state of the other subsystem, regardless of distance. This phenomenon is described mathematically by the non-local correlations predicted by quantum mechanics.

\section{Unitary Criteria for Quantum Context-free Grammars}
  In this section we try to establish a unitary condition on the quantum context-free grammars.

\begin{definition}
     For each $a \in T$, let $U_a$ be the evolution matrix whose rows are indexed by sentential forms containing at least one nonterminal and whose columns correspond to the ordered pair consisting of sentential forms containing at least one occurrence of $a$ together with the rule used to derive it. The $ij^{th}$  entry of this matrix is the amplitude of the production $A \rightarrow \beta$  where $xAy$  indexes the $i^{th}$ row and $(x\beta y, A \rightarrow \beta)  $ indexes the $j^{th}$ column. If there is no such derivation then the $ij^{th}$ entry is $0$.
     \end{definition}
     \begin{definition}
    A quantum context-free grammar is unitary if and only if each of its  evolution matrices $U_a$ as defined above is unitary. 
\end{definition}
\par Let $\mathcal{S} $ denote the set of all sentential forms generated by a quantum context-free grammar. Unlike classical context-free grammars, where each step produces a specific sentential form, quantum grammars operate with superpositions of sentential forms. 
Each sentential form corresponds to a basis  vector in the Hilbert space $l_2(\mathcal{S})$. Since a single nonterminal can be replaced by a linear combination of these sentential forms, 
a global production of the quantum grammar in the space $l_2(\mathcal{S})$, which corresponds to a superposition  of the productions with the same nonterminal on the left, has the form $\ket{\psi}= \sum\limits_{s \in \mathcal{S}}\alpha_s\ket{s}$ with 
we have $\sum\limits_{s \in \mathcal{S}}|\alpha_s|^2 = 1$ where $\alpha_s \in \mathbb{C}$ is the amplitude of the sentential form. 
 The next step of the further derivation can be obtained by the application of the linear operator $U_a$  \\
 \\
\textbf {Well-formedness conditions for quantum context-free grammars} \\\\
Since the matrix $U_a$ is infinite  we cannot directly apply definition \ref{Uni} to verify that it is unitary. Hence we propose certain criteria to ensure unitarity. The following well-formedness conditions are derived by observing that, from any given sentential form of the quantum grammar, it is feasible to transition to only a finite number of other sentential forms within a single step. Conversely, the given sentential form can also be reached from only a finite number of distinct sentential forms in one step. 
\begin{enumerate}
    \item \label{C1} Row vector norm condition:
    \[\sum_{\beta \in R(A) }|\textbf{c}(A, \beta)|^2=1\]
    
    \item  Separability Conditions:
    \begin{enumerate}
        \item \label{C2}$\sum\limits_{x\beta_1y~=~u\beta_2v ~  \in~ (N \cup T)^*  } \bar{\mathbf{c}}(A_1,\beta_1)\cdot \mathbf{c}(A_2,\beta_2)=0 $\\
     \hspace*{8cm}where $\beta_i \in R(A_i) $ 

    \item \label{C3} $\sum\limits_{|\beta_1| \neq |\beta_2|}\bar{\mathbf{c}}(A, 
    \beta_1)\cdot \mathbf{c}(A,\beta_2)=0$ \\ 
   \hspace*{8cm}where $\beta_i \in R(A)$  
    \end{enumerate}
\end{enumerate}

\begin{thm}
    The quantum context-free grammar is unitary if and only if the well-formedness conditions are satisfied.
\end{thm}
\begin{proof}
To ensure unitarity, we consider two distinct columns of the infinite evolution matrix and compute their inner product within the Hilbert space. Unitarity demands that the columns are orthogonal, thereby preserving the orthogonality condition inherent to unitary operators. Despite the columns being infinite, only a finite number of non-zero summands contribute to the sum representing their inner product. Consequently, the resulting sum is finite and must equal zero to satisfy the orthogonality condition. From the point of quantum mechanics this means that two derivations of a single sentential form  which are entangled must be identical. In other words separability implies that if a sentential form can  be reached by two different paths then the last production in the two derivations should be orthogonal. Moreover two sentential forms reachable form the same sentential form using productions of varying lengths  cannot be entangled. 
\par Given $G = (N, T, I, P )$ a quantum context-free grammar satisfying the above conditions, Condition \ref{C1} ensures that the rows of the QCFG  evolution matrix  are normalised and Condition \ref{C2} ensures that  every pair of columns indexed by sentential forms $x \beta y$ are orthogonal, but with different rules while Condition \ref{C3} ensures that column vectors of sentential forms with different lengths are orthogonal.\\
      Clearly these three conditions ensure that $U_a$ is unitary.
     \end{proof}
\begin{exmp}
   Let $G=(\{I,A,B\},\{a,b\},I,P)$ be a quantum context-free \cite{golovkins2000quantum} of dimension 4  where \[P= \{I \rightarrow aIB, I\rightarrow aB, I \rightarrow aABB, A\rightarrow a, B\rightarrow b\}\] 
   \[\text{with  amplitudes }~~~ \Big< \frac{1}{2\sqrt{3}}, \frac{1}{2\sqrt{3}}, \frac{1}{2\sqrt{3}}, \frac{1}{2\sqrt{3}}\Big > , ~ \Big< \frac{1}{2\sqrt{3}}, \frac{-1}{2\sqrt{3}}, \frac{1}{2\sqrt{3}}, \frac{-1}{2\sqrt{3}}\Big >,\]\[ \Big< \frac{1}{2\sqrt{3}}, \frac{1}{2\sqrt{3}}, \frac{-1}{2\sqrt{3}}, \frac{-1}{2\sqrt{3}}\Big> ,~
     \Big< \frac{1}{2}, \frac{-1}{2}, \frac{-1}{2}, \frac{1}{2}\Big > \text{ and } \Big< \frac{i}{2}, \frac{-i}{2}, \frac{-i}{2}, \frac{i}{2}\Big > ~\text{respectively}. \]\\
Consider the derivation
\begin{align*}
I &\implies aIB ~~with~ amplitude &  \Big< \frac{1}{2\sqrt{3}}, \frac{1}{2\sqrt{3}}, \frac{1}{2\sqrt{3}}, \frac{1}{2\sqrt{3}}\Big >\\
&\implies aaIBB ~~with~ amplitude & \Big< \frac{1}{(2\sqrt{3})^2}, \frac{1}{(2\sqrt{3})^2}, \frac{1}{(2\sqrt{3})^2}, \frac{1}{(2\sqrt{3})^2}\Big >\\
&\vdots &\vdots\\
&\implies a^{n-1}IB^{n-1} ~~with ~amplitude & \Big< \frac{1}{(2\sqrt{3})^{n-1}}, \frac{1}{(2\sqrt{3})^{n-1}}, \frac{1}{(2\sqrt{3})^{n-1}}, \frac{1}{(2\sqrt{3})^{n-1}}\Big > \\
&\implies a^{n}B^{n} ~~with ~amplitude &\Big< \frac{1}{(2\sqrt{3})^n}, -\frac{1}{(2\sqrt{3})^n}, \frac{1}{(2\sqrt{3})^n}, -\frac{1}{(2\sqrt{3})^n}\Big >\\
&\implies a^{n}bB^{n-1} ~~with ~amplitude &\Big< \frac{i}{2}\cdot\frac{1}{(2\sqrt{3})^n}, -\frac{-i}{2}\cdot\frac{1}{(2\sqrt{3})^n}, \frac{-i}{2}\cdot\frac{1}{(2\sqrt{3})^n}, -\frac{i}{2}\cdot\frac{1}{(2\sqrt{3})^n}\Big >\\
&\vdots&\vdots\\
\end{align*}
\[\implies a^{n}b^{n} ~with ~amplitude~ \Big< \Big(\frac{i}{2}\Big)^n\frac{1}{(2\sqrt{3})^n}, -\Big(\frac{-i}{2}\Big)^n\frac{1}{(2\sqrt{3})^n},\Big(\frac{-i}{2}\Big)^n\frac{1}{(2\sqrt{3})^n},-\Big(\frac{i}{2}\Big)^n\frac{1}{(2\sqrt{3})^n}\Big >\]

Also we have another derivation
\begin{align*}
I &\implies aIB ~~with~ amplitude &\Big< \frac{1}{2\sqrt{3}}, \frac{1}{2\sqrt{3}}, \frac{1}{2\sqrt{3}}, \frac{1}{2\sqrt{3}}\Big >\\
&\implies aaIBB ~~with~ amplitude & \Big< \frac{1}{(2\sqrt{3})^2}, \frac{1}{(2\sqrt{3})^2}, \frac{1}{(2\sqrt{3})^2}, \frac{1}{(2\sqrt{3})^2}\Big >\\
&\vdots\\
&\implies a^{n-2}IB^{n-2} ~~with ~amplitude & \Big< \frac{1}{(2\sqrt{3})^{n-2}}, \frac{1}{(2\sqrt{3})^{n-2}}, \frac{1}{(2\sqrt{3})^{n-2}}, \frac{1}{(2\sqrt{3})^{n-2}}\Big > \\
&\implies a^{n-1}AB^{n} ~~with ~amplitude & \Big< \frac{1}{(2\sqrt{3})^{n-1}}, \frac{1}{(2\sqrt{3})^{n-1}}, -\frac{1}{(2\sqrt{3})^{n-1}}, -\frac{1}{(2\sqrt{3})^{n-1}}\Big > \\
&\implies a^{n}B^{n} ~~with ~amplitude &\Big<\frac{1}{2}\cdot \frac{1}{(2\sqrt{3})^{n-1}}, -\frac{1}{2}\cdot \frac{1}{(2\sqrt{3})^{n-1}}, \frac{1}{2}\cdot \frac{1}{(2\sqrt{3})^{n-1}},\\ &&-\frac{1}{2}\cdot \frac{1}{(2\sqrt{3})^{n-1}}\Big >\\
\end{align*}
\begin{align*}
&\implies a^{n}bB^{n-1} ~~with ~amplitude &\Big< \frac{i}{2}\cdot\frac{1}{2}\cdot\frac{1}{(2\sqrt{3})^{n-1}}, \frac{-i}{2}\cdot\frac{1}{2}\cdot\frac{1}{(2\sqrt{3})^{n-1}}, \frac{-i}{2}\cdot\frac{1}{2}\cdot\frac{1}{(2\sqrt{3})^{n-1}},\\ && -\frac{i}{2}\cdot\frac{1}{2}\cdot\frac{1}{(2\sqrt{3})^{n-1}}\Big >\\
 &\vdots\\
 \end{align*}
 \begin{align*}
&\implies a^{n}b^{n} ~~with ~amplitude &\hspace{2.5 cm}\Big< \Big(\frac{i}{2}\Big)^n\frac{1}{2}\cdot\frac{1}{(2\sqrt{3})^{n-1}}, -\Big(\frac{-i}{2}\Big)^n\frac{1}{2}\cdot\frac{1}{(2\sqrt{3})^{n-1}},\\ &&\Big(\frac{-i}{2}\Big)^n\frac{1}{2}\cdot\frac{1}{(2\sqrt{3})^{n-1}}, -\Big(\frac{i}{2}\Big)^n\frac{1}{2}\cdot\frac{1}{(2\sqrt{3})^{n-1}}\Big >\\
\end{align*}
Hence the amplitude of the word $a^nb^n $ generated by the quantum grammar $G$ is \[ \Big<\Big(\frac{i}{2}\Big)^n\frac{1+\sqrt{3}}{(2\sqrt{3})^n},-\Big(\frac{-i}{2}\Big)^n\frac{1+\sqrt{3}}{(2\sqrt{3})^n},\Big(\frac{-i}{2}\Big)^n\frac{1+\sqrt{3}}{(2\sqrt{3})^n}\Big)^n,-\Big(\frac{i}{2}\Big)^n\frac{1+\sqrt{3}}{(2\sqrt{3})^n}\Big >\]
and the probability associated with it is \[f(a^{n}b^{n})=\frac{8(2+\sqrt{3})}{48^n} <1, \text{ for } n>1\] 
For $n=1$ we have only one possible derivation 
\begin{align*}
 I&\implies aB\hspace{2 cm}\Big<\frac{1}{2\sqrt{3}},\frac{-1}{2\sqrt{3}}, \frac{1}{2\sqrt{3}} , \frac{-1}{2\sqrt{3}}\Big>  \\
 &\implies ab \hspace{2 cm}\Big<\frac{i}{4\sqrt{3}},\frac{i}{4\sqrt{3}}, \frac{-i}{4\sqrt{3}} , \frac{-i}{4\sqrt{3}}\Big> 
\end{align*}

which produces a word $ab$ with probability
\[f(ab)=\frac{1}{12}\]
Hence the language is well-formed.

\end{exmp}

\section{Conclusion}
In this paper we have attempted to address the open problem (7) posed by Moore and Crutchfield in \cite{moore2000quantum} and that posed by \cite{qiu2002quantum}. It may be observed that the condition discussed here corresponds to  the quantum pushdown automata constructed by Golovkins in \cite{golovkins2000quantum}. 

\bibliographystyle{unsrt}  
\bibliography{references}

\end{document}